\documentclass[10pt,journal,a4paper]{IEEEtran}%TWC
\usepackage{amssymb}
\usepackage{amsmath}
\usepackage{bm}
\usepackage{amsthm}
\usepackage{graphicx}
\usepackage{subfigure}
\usepackage{float}
\usepackage{cite}
\usepackage{enumerate}
\usepackage{enumitem}
\usepackage{multirow}
\usepackage{amsfonts}
\usepackage{color, soul}
\usepackage{url}
\usepackage{amsthm}

\begin{document}
\newtheorem{remark}{\bf~~Remark}
\newtheorem{proposition}{\bf~~Proposition}
\newtheorem{theorem}{\bf~~Theorem}
\newtheorem{definition}{\bf~~Definition}
\title{\LARGE Reconfigurable Intelligent Surface~(RIS) Assisted Wireless Coverage Extension: RIS Orientation and Location Optimization}

\author{
\IEEEauthorblockN{
	{Shuhao Zeng}, \IEEEmembership{Student Member, IEEE},
	{Hongliang Zhang}, \IEEEmembership{Member, IEEE},
	{Boya Di}, \IEEEmembership{Member, IEEE},\\
	{Zhu Han}, \IEEEmembership{Fellow, IEEE},
	{and Lingyang Song}, \IEEEmembership{Fellow, IEEE}}
	\vspace{-0.5cm}
	\thanks{Manuscript received June 17, 2020; revised August 19, 2020; accepted September 14, 2020. This work was supported in part by the National Natural Science Foundation of China under Grants 61625101, 61829101, and 61941101, and in part by NSF EARS-1839818, CNS-1717454, CNS-1731424, and CNS-1702850. The associate editor coordinating the review of this paper and approving it for publication was G. Alexandropoulos. \emph{(Corresponding author: Lingyang Song.)}}
	\thanks{S. Zeng and L. Song are with Department of Electronics, Peking University, Beijing, China (email: \{shuhao.zeng, lingyang.song\}@pku.edu.cn).}
	\thanks{H. Zhang is with Department of Electronics, Peking University, Beijing, China, and also with Department of Electrical Engineering, Princeton University, USA (email: hongliang.zhang92@gmail.com).}
	\thanks{B. Di is with Department of Electronics, Peking University, Beijing, China, and also with Department of Computing, Imperial College London, UK (email: diboya92@gmail.com).}
	\thanks{Z. Han is with Electrical and Computer Engineering Department, University of Houston, Houston, TX, USA, and also with the Department of Computer Science and Engineering, Kyung Hee University, Seoul, South Korea (email: zhan2@uh.edu).}
}

%\IEEEauthorblockA{\small{School of Electronics Engineering and Computer Science, Peking University, Beijing, China.} \\
%\small{Email: \{zhangshuhang, hongliang.zhang, qichenhe, bkg, lingyang.song\}@pku.edu.cn}}}

\maketitle

\begin{abstract}
Recently, reconfigurable intelligent surfaces~(RIS) have attracted a lot of attention due to their capability of extending cell coverage by reflecting signals toward the receiver. In this letter, we analyze the coverage of a downlink RIS-assisted network with one base station~(BS) and one user equipment~(UE). Since the RIS orientation and the horizontal distance between the RIS and the BS have a significant influence on the cell coverage, we formulate an RIS placement optimization problem to maximize the cell coverage by optimizing the RIS orientation and horizontal distance. To solve the formulated problem, a coverage maximization algorithm~(CMA) is proposed, where a closed-form optimal RIS orientation is obtained. Numerical results verify our analysis. 
%and show the influence of Rician factor and RIS size on the cell coverage.
\end{abstract}
\vspace{-0.1cm}
%\vspace{-0.5cm}
\section{Introduction}
\vspace{-0.1cm}
%%%%%%%%%%%%%%%%%%%%%%%%%%%%%%%%%%%%%%%%%%%%%%%
Reconfigurable intelligent surfaces~(RIS) are an emerging wireless communication technology and is capable of shaping the propagation environment into a desired form~\cite{MHLKZG,C_2018}. The RIS is inlaid with a large number of reflecting elements whose electromagnetic~(EM) response can be tuned dynamically with onboard positive-intrinsic-negative~(PIN) diodes~\cite{M_2019}. By controlling the EM responses of the elements, the received signal strength can be optimized~\cite{CSG_2020}.

In the literature, RIS-assisted wireless networks have been studied to increase coverage and improve link quality. Existing works on RIS-assisted networks can be roughly divided into point-to-point~\cite{LYMM_2020,XDR_2019} and multi-user networks~\cite{BHLYZH_2020,CAGMC_2019}. In~\cite{LYMM_2020}, the authors considered an RIS-assisted point-to-point network without direct link, where the network coverage, the signal-to-noise-ratio~(SNR) gain, and the delay outage rate of this network were analyzed. The authors in~\cite{XDR_2019} maximized the sum rate of a point-to-point RIS-assisted multi-input single-output system by jointly optimizing the beamformer at the transmitter and continuous phase shifts of the RIS. In~\cite{BHLYZH_2020}, a multi-user network assisted by an RIS was investigated, where the digital beamformer and the RIS configuration are jointly optimized to maximize the sum rate. In~\cite{CAGMC_2019}, the authors investigated a multi-user RIS-assisted network, where the transmit power allocation and the phase shifts of the RIS are optimized to improve energy efficiency.

However, existing works only utilized the RIS for coverage extension given the RIS location, but how to deploy the RIS to further maximize the cell coverage has not been studied yet. In this letter, we consider an RIS-assisted downlink cellular network with one base station~(BS) and one user equipment~(UE). The cell coverage of this network is analyzed first. Then, since the RIS orientation and the horizontal distance between the RIS and the BS will significantly influence the cell coverage, we maximize the cell coverage by optimizing the RIS orientation and horizontal distance. To solve the RIS placement optimization problem, we propose a coverage maximization algorithm~(CMA), where a closed-form optimal RIS orientation is derived first, and the horizontal distance is then optimized using the interior point method. 

The rest of this letter is organized as follows. In Section~\ref{sec_sys_model}, the system model for the RIS-assisted network is introduced. The cell coverage is analyzed in Section~\ref{sec_coverage_ana}. In Section~\ref{sec_placement}, we first formulate a coverage maximization problem, and then propose the CMA to solve the problem. In Section~\ref{sec_simulation}, simulation results validate our analysis. Finally, we conclude our work in Section~\ref{sec_conclusion}.
%\vspace{-0.4cm}
\section{System Model}%
\vspace{-0.1cm}
\label{sec_sys_model}
In this section, we first introduce the RIS-assisted network model. The channel model is then constructed. 
\vspace{-0.2cm}
\subsection{Scenario Description}
%\vspace{-0.1cm}
As shown in Fig.~\ref{system_model}, we consider a narrow-band downlink network with one UE and one BS. Due to the dynamic wireless environment involving unexpected fading, the link between the UE and the BS can be unstable or even fall into a complete outage~\cite{HBLZ_2020}. To solve this problem, we introduce an RIS to assist the communication. To describe the topology of the system, a Cartesian coordinate is adopted, where the $xoy$ plane coincides with the RIS surface, and the $z$-axis is vertical to the RIS. Based on the $xoy$ plane, the space can be divided into two sides. To ensure that the RIS can reflect the signal from the BS towards the UE, we assume that the UE and the BS are in the same side of the RIS, i.e., $z>0$ in the coordinate system. 
%\vspace{-0.2cm}
{\setlength{\abovecaptionskip}{0.1cm}
	\setlength{\belowcaptionskip}{0cm} 
\begin{figure}[!tpb]
	\centering
	\includegraphics[width=2.5in]{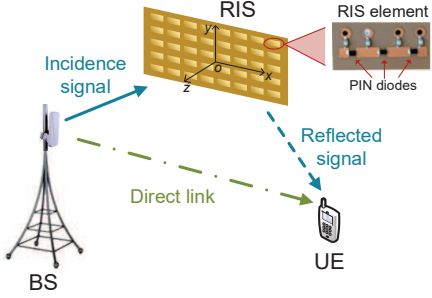}
	\vspace{-0.4cm}
	\caption{System model of the RIS-assisted cellular network.}
	\vspace{-0.6cm}
	\label{system_model}
\end{figure}
}

The RIS is composed of $M\times N$ sub-wavelength elements, each with the size of $s_M\times s_N$. As shown in Fig.~\ref{system_model}, the RIS element contains several PIN diodes~\cite{LZRJXXQT_2020}. When the biased voltages applied to the PIN diodes change, the RIS element will change the reflection phase shift accordingly. %Therefore, the achievable phase shifts are discrete and limited~\cite{TDR_2010}. 
%In this paper, it is assumed that the RIS is $L$ bits coded, i.e., $2^L$ patterns of phase shifts with uniform interval $\Delta \varphi=\frac{2\pi}{2^L}$ can be generated via controlling the states of the PIN diodes. 
Define $\Gamma$ and $\varphi^{m,n}$ as the reflection amplitude change and phase shift of the $(m,n)$-th element, respectively\footnote{Without loss of generality, we assume that the states of the PIN diodes in an element do not influence reflection amplitude change $\Gamma$ of the element.}. The reflection coefficient of the $(m,n)$-th element can thus be written as $\Gamma_{m,n}=\Gamma e^{-j\varphi_{m,n}}$. Assume that the BS is in the far field of the RIS.  Amplitude change $\Gamma$ can be modeled by $\cos\theta_i$~\cite{OEE_2020}, where $\theta_i$ is the incidence angle from the BS to the RIS. Besides, we assume that phase shift $\varphi_{m,n}$ is not influenced by the incidence and reflection angles.
%\begin{align}
%\Gamma_{m,n}=\Gamma e^{-j\varphi_{m,n}}.
%\end{align}
%\vspace{-0.4cm}
\subsection{Channel Model}
%\vspace{-0.1cm}
The channel from the BS to the UE is composed of $MN$ RIS-based channels and a direct link, where the $(m,n)$-th RIS-based channel represents the channel from the BS to the UE via the $(m,n)$-th RIS element\footnote{Since we consider average performance here, the small scale fading of the RIS-based channel is averaged. Besides, it is assumed that the small-scale fading corresponding to different RIS elements are independently distributed.}.

%Although small scale fading will influence the cell coverage, we only consider pathloss here, since the small scale fading renders it difficult to maximize the cell coverage by optimizing RIS deployment.

The channel gain for the $(m,n)$-th RIS-based channel can be expressed as~\cite{HBLZ_2020},
%{\setlength\abovedisplayskip{0cm}
%\setlength\belowdisplayskip{0.05cm}
\begin{align}
\label{RIS_based}
h_{m,n}=\frac{\lambda\sqrt{Gs_Ms_N}}{(4\pi)^{\frac{3}{2}}\sqrt{D_{m,n}^{\alpha}d_{m,n}^{\alpha}}}e^{-j\frac{2\pi}{\lambda}(D_{m,n}+d_{m,n})},
\end{align}
where $\lambda$ is the wavelength corresponding to the carrier frequency, $G$ is antenna gain, $s_M$ and $s_N$ are the size of one RIS element, $\alpha$ represents the pathloss exponent, $D_{m,n}$ and $d_{m,n}$ are the distance between the BS and the $(m,n)$-th RIS element, and the distance between the $(m,n)$-th RIS element and the UE, respectively. Since the BS is far away from the RIS, its distances from different RIS elements are approximately the same, i.e., $D_{m,n}=D$, where $D$ is the distance between the BS and the center of the RIS. Besides, in order to derive the area of the cell coverage, we focus on finding the cell edge, and thus, it can be assumed that the UE is in the far field of the RIS. Therefore, we have $d_{m,n}=d$, where $d$ is the distance between the UE and the center of the RIS. As a result, the pathloss is common to all RIS elements, i.e., $GD_{m,n}^{-\alpha}d_{m,n}^{-\alpha}=GD^{-\alpha}d^{-\alpha}\triangleq PL_R$
%\begin{align}
%GD_{m,n}^{-\alpha}d_{m,n}^{-\alpha}=PL_{LoS},\\
%PL(D_{m,n})PL(d_{m,n})=PL_{NLoS}.
%\end{align}

For the direct link, the channel model can be written as
%{\setlength\abovedisplayskip{0cm}
%\setlength\belowdisplayskip{0cm}
\begin{align}
\label{direct_link}
h_D=\frac{\lambda\sqrt{G}}{4\pi d_{BU}}e^{-j\frac{2\pi}{\lambda}d_{BU}},
\end{align}
where $d_{BU}$ represents the distance between the BS and the UE. Based on (\ref{RIS_based}) and (\ref{direct_link}), the channel from the BS to the UE is given by
%{\setlength\abovedisplayskip{-0.2cm}
%	\setlength\belowdisplayskip{0cm}
\begin{align}
\label{channel_A}
h=\sum_{m,n}\Gamma_{m,n}h_{m,n}+h_D,
\end{align}
%}
Therefore, the received signal-to-noise ratio~(SNR) can be expressed as $\gamma=\frac{P|h|^2}{\sigma^2}$,
%\begin{align}
%\gamma=\frac{P|h|^2}{\sigma^2},
%\end{align}
where $P$ denotes the transmit power of the BS, and $\sigma^2$ is the variance of the additive white Gaussian noise~(AWGN) received at the UE.

%where $\eta_{NLoS}=\frac{P\Gamma^2}{\sigma^2}(PL_{NLoS})^2$, $\eta_{LoS}=\frac{P\Gamma^2}{\sigma^2}PL_{LoS}$, $\eta_X=\frac{P\lambda\sqrt{G}}{\sigma^24\pi d_{BU}}$, and $\eta_D=\frac{P\lambda^2G}{\sigma^2(4\pi)^2(d_{BU})^2}$

%\vspace{-0.3cm}
 
\section{Cell Coverage Analysis}
%\vspace{-0.1cm}
\label{sec_coverage_ana}
Before analyzing the cell coverage, we first introduce the definition.
\begin{definition}
	The cell coverage is defined as an area where the received SNR at the UE is larger than a certain threshold $\gamma_{th}$~\cite{JDMT_2003,HPJX_2019}, i.e.,
	{\setlength\abovedisplayskip{-0.2cm}
		\begin{align}
		\label{def_coverage}
		\gamma \ge \gamma_{th}=\gamma_{s}L_{mar},
		\end{align}
		where $\gamma_s$ represents UE sensitivity, and $L_{mar}$ is margin for penetration loss.
	}
\end{definition}
%\vspace{-0.2cm}

In the following, we first find the optimal phase shifts of the RIS that maximizes the SNR. Then, the cell coverage is analyzed.
%\vspace{-0.5cm}
\subsection{Optimal Phase Shifts of RIS}
%\vspace{-0.1cm}
To improve the system performance, the phase shifts are optimized to maximize the SNR. Define $\eta_R=\frac{P}{\sigma^2}\frac{\lambda^2}{(4\pi)^3}\big(\cos^2(\theta_i) s_Ms_N\big)$, $\eta_D=\frac{P}{\sigma^2}\frac{\lambda^2 G}{(4\pi)^2}$, and  $\eta_X=2\frac{P}{\sigma^2}\frac{\lambda^2\sqrt{G}}{(4\pi)^{\frac{5}{2}}}\big(\cos(\theta_i)\sqrt{s_Ms_N}\big)$. The maximum SNR can be derived, given in the following theorem.
% Optimal phase shifts and the maximum SNR are derived in the following theorem, where it is assumed that the phase shifts take continuous values.
\begin{theorem}
	\label{E_A}
	%\vspace{-0.2cm}
	When the phase shifts are set as
%	{\setlength\abovedisplayskip{0cm}
%		\setlength\belowdisplayskip{0cm}
	\begin{align}
	\label{opt_phase}
	\varphi_{m,n}=\big(\frac{2\pi}{\lambda}d_{BU}-\frac{2\pi}{\lambda}(D_{m,n}+d_{m,n})\big) \mod 2\pi,
	\end{align}
	the SNR is maximized, and its value can be given by
	\begin{align}
	\label{max_SNR_cont}
	\gamma=\eta_R M^2N^2PL_R+\eta_Dd_{BU}^{-2}+\eta_XMN\frac{\sqrt{PL_R}}{d_{BU}}.
	\end{align}
%}
\end{theorem}
\begin{proof}
	%\vspace{-0.1cm}
	Based on (\ref{channel_A}), the SNR can be rewritten as
%	{\setlength\abovedisplayskip{0cm}
%		\setlength\belowdisplayskip{0cm}
%		\begin{footnotesize}
			\begin{align}
			\gamma=&\eta_{R}PL_{R}\sum_{m,m',n,n'}e^{j(\varphi_{m,n}+\theta_{m,n}-\varphi_{m',n'}-\theta_{m',n'})}+\eta_Dd_{BU}^{-2}\notag\\
			&+\eta_X\frac{\sqrt{PL_{R}}}{d_{BU}}\sum_{m,n}\cos(\varphi_{m,n}+\theta_{m,n}-\frac{2\pi}{\lambda}d_{BU}),
			%E(\gamma)=&\frac{\eta_{NLoS}MN}{\kappa+1}+\frac{\eta_{LoS}\kappa}{\kappa+1}\sum_{m,m',n,n'}e^{j(\phi_{m,n}+\theta_{m,n}-\phi_{m',n'}-\theta_{m',n'})}\notag\\
			%&+\sum_{m,n}2\eta_X\sqrt{\frac{\kappa}{\kappa+1}}\cos(\varphi_{m,n}+\theta_{m,n}-\frac{2\pi}{\lambda}d_{BU})+\eta_D,
			\end{align}
			where $\theta_{m,n}=\frac{2\pi}{\lambda}(D_{m,n}+d_{m,n})$.
%		\end{footnotesize}
		We can find that when the phase shifts are set as shown in (\ref{opt_phase}), $\sum_{m,m',n,n'}e^{j(\phi_{m,n}+\theta_{m,n}-\phi_{m',n'}-\theta_{m',n'})}$ and $\sum_{m,n}\cos(\varphi_{m,n}+\theta_{m,n}-\frac{2\pi}{\lambda}d_{BU})$ are simultaneously maximized, with the values being $M^2N^2$ and $MN$, respectively~\cite{HBLZ_2020}. Therefore, Theorem~\ref{E_A} holds.
	%\vspace{-0.4cm}
\end{proof}
\subsection{Cell Coverage Analysis}
%\vspace{-0.1cm}
\begin{figure}[!tpb]
	\centering
	\includegraphics[width=2in]{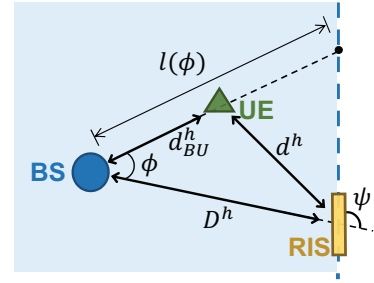}
	%\vspace{-0.3cm} 
	\caption{Top view of the RIS-assisted network.}
	%\vspace{-0.6cm}
	%	\setlength{\belowcaptionskip}{0cm}
	\label{system_profile}
\end{figure}
In this section, the cell coverage in a given direction is investigated first. The area of the cell coverage is then derived. 
\subsubsection{Cell Coverage in a given Direction}
As shown in Fig.~\ref{system_profile}, we use $\phi$ to represent the angle between the given direction and the direction from the BS to the RIS, where $\phi \in [0,2\pi)$. To derive the cell coverage in direction $\phi$, we first analyze how SNR $\gamma$ changes with the horizontal distance between the BS and the UE, denoted by $d_{BU}^h$. 

When $d_{BU}^h$ is shorter than the horizontal distance between the BS and the RIS, denoted by $D^h$, it is hard to find out the influence of $d_{BU}^h$ on SNR $\gamma$. This is because when $d_{BU}^h$ becomes larger, the horizontal distance between the UE and the RIS, denoted by $d^h$, first decreases and then increases. To avoid discussing the trend of $\gamma$ with respect to $d_{BU}^h$, we require that SNR $\gamma$ is always above threshold $\gamma_{th}$ when $d_{BU}^h<D^h$. Note that $\gamma$ is lower bounded by $\eta_Dd_{BU}^{-2}$, which decreases as the UE moves away from the BS. Therefore, by assuming that when $d_{BU}^h=D^h$, $\eta_Dd_{BU}^{-2}$ does not fall below $\gamma_{th}$, i.e., 
%{\setlength\abovedisplayskip{0cm}
%\setlength\belowdisplayskip{0cm}
\begin{align}
\label{above_thr}
\eta_D \Big(\sqrt{(D^h)^2+(H_B-H_U)^2}\Big)^{-2} \ge \gamma_{th},
\end{align}
where $H_B$ and $H_U$ are the height of the BS and the UE, respectively, the SNR is guaranteed to exceed $\gamma_{th}$.

When $d_{BU}^h$ exceeds $D^h$, it is obvious that $\gamma$ strictly decreases as the UE moves away from the BS. Besides, $\gamma$ approaches to $0$ when the UE is sufficiently far away from the BS. 

%\begin{proof}
%	\vspace{-0.2cm}
%	See Appendix \ref{app_region_A}.
%	\vspace{-0.3cm}
%\end{proof} 
Based on above discussions, it can be found that the horizontal distance $d_{BU}^h$ with $\gamma=\gamma_{th}$ is unique, which is denoted by $d_{th}(\phi)$. Define $g$ as the function of $\gamma$ with respect to the angle $\phi$ and the horizontal distance $d_{BU}^h$, i.e., $\gamma=g(\phi,d_{BU}^h)$. We then have
%{\setlength\abovedisplayskip{0cm}
%	\setlength\belowdisplayskip{0cm}
\begin{align}
\label{d_th}
g\big(\phi,d_{th}(\phi)\big)=\gamma_{th}, \forall \phi.
\end{align}
The expression of $g$ is given by
\begin{align}
\label{expression_g}
g(\phi,d_{BU}^h)&=\eta_R M^2N^2PL_R(\phi,d_{BU}^h)+\eta_D\Big(d_{BU}(\phi,d_{BU}^h)\Big)^{-2}\notag\\
&+\eta_XMN\sqrt{PL_R(\phi,d_{BU}^h)}/d_{BU}(\phi,d_{BU}^h),
\end{align}
where $PL_R(\phi,d_{BU}^h)=GD^{-\alpha}\Big((D^h-d_{BU}^h\cos\phi)^2+(d_{BU}^h\sin\phi)^2+(H_R-H_U)^2\Big)^{-\frac{\alpha}{2}}$, and $d_{BU}(\phi,d_{BU}^h)=\sqrt{(d_{BU}^h)^2+(H_B-H_U)^2}$. We can also obtain that SNR $\gamma$ exceeds threshold $\gamma_{th}$ only if $d_{BU}^h \le d_{th}(\phi)$. 

Recall that the UE and the BS are in the same side of the RIS. Therefore, the cell coverage in some directions cannot achieve $d_{th}(\phi)$. Specifically, define $\psi$ as the angle between the RIS orientation and direction from the BS to the RIS, where $\psi \in (0,\pi)$, and we have the following two cases,
\begin{enumerate}[label=$\bullet$]
	\item $\phi \in [0,\psi)\cup(\psi+\pi,2\pi)$: To guarantee that the UE and the BS are in the same side of the RIS, the horizontal distance between the BS and the UE has to satisfy
	{\setlength\abovedisplayskip{0cm}
		\setlength\belowdisplayskip{0cm}
	\begin{align}
	d_{BU}^h \le l(\phi),
	\end{align}
	where $l(\phi)=D^h/(\cos(\phi)-\sin(\phi)\cot(\psi))$ denotes the distance between the BS and the RIS in direction $\phi$, as shown in Fig.~\ref{system_profile}. Therefore, the cell coverage is given by
	\begin{align}
	\label{coverage_phi_min}
	c(\phi)=\min\big(d_{th}(\phi),l(\phi)\big).
	\end{align}
	\item $\phi \in [\psi,\psi+\pi]$: The UE and the BS are always in the same side of the RIS. Therefore, the cell coverage is
	\begin{align}
	\label{coverage_phi}
	c(\phi)=d_{th}(\phi).
	\end{align}
}
\end{enumerate} 
To have a better understanding of the cell coverage in directions $\phi \in [0,\psi)\cup(\psi+\pi,2\pi)$, we compare $d_{th}(\phi)$ with $l(\phi)$ in the following remark.
\begin{remark}
	%\vspace{-0.2cm}
	\label{profile}
	When $\phi \in [0,\psi)\cup(\psi+\pi,2\pi)$, there are two angles satisfying $d_{th}(\phi)=l(\phi)$, denoted by $\phi_{u}$ and $\phi_{l}$, i.e.,
	{\setlength\abovedisplayskip{0cm}
		\setlength\belowdisplayskip{0cm}
		\begin{align}
		\label{cons_phi_l}
		d_{th}(\phi_l)=l(\phi_l), \phi_l <\pi,\\
		\label{cons_phi_u}
		d_{th}(\phi_u)=l(\phi_u), \phi_u >\pi.
		\end{align}
	Besides, $d_{th}(\phi)>l(\phi)$ only when $\phi \in [0,\phi_l)\cup(\phi_u,2\pi)$.}
%	When $\phi \in [\mu,\mu+\pi]$, the cell coverage $c(\phi)$ decreases as $\phi$ approaches $\pi$. Besides, the direction with the largest cell coverage, denoted by $\phi^*$, satisfies 
%	\begin{align}
%	\label{opt_dir}
%	d_A(\phi^*)=\frac{D^h}{\cos\phi^*-\sin \phi^* \cot \mu }.
%	\end{align}
\end{remark}
\begin{proof}
	%\vspace{-0.2cm}
	See Appendix \ref{app_profile}.
	%\vspace{-0.2cm}
\end{proof}
\subsubsection{Area of cell coverage}
Based on (\ref{coverage_phi_min}),  (\ref{coverage_phi}), and Remark~\ref{profile}, we can derive the area of the cell coverage, given in the following theorem.
\begin{theorem}
	%\vspace{-0.2cm}
	The area of the cell coverage is given by
%	{\setlength\abovedisplayskip{0cm}
%		\setlength\belowdisplayskip{-0.2cm}
	\begin{align}
	\label{area}
	S=\int_{\phi_l}^{\phi_u} \frac{1}{2}d_{th}^2(\phi)d\phi+\frac{1}{2}\sin(\phi_l-\phi_u)l(\phi_l)l(\phi_u).
	\end{align}
%}
\end{theorem}
\begin{proof}
	The area of the cell coverage can be expressed as $S=\int_{0}^{2\pi} \frac{c^2(\phi)}{2}d\phi$.
%	\begin{align}
%	S=\int_{0}^{2\pi} \frac{1}{2}c^2(\phi)d\phi.
%	\end{align}
	According to Remark~\ref{profile}, when $\phi \in [0,\phi_l)\cup(\phi_u,2\pi)$, $c(\phi)=l(\phi)$. Otherwise, $c(\phi)=d_{th}(\phi)$. Therefore, we have
	{\setlength\abovedisplayskip{0.05cm}
		\setlength\belowdisplayskip{0.05cm}
	\begin{small}
	\begin{align}
	\label{area_proof_1}
	S=\int_{\phi_l}^{\phi_u} \frac{d_{th}^2(\phi)}{2}d\phi+\int_{0}^{\phi_l} \frac{l^2(\phi)}{2}d\phi+\int_{\phi_u}^{2\pi} \frac{l^2(\phi)}{2}d\phi.
	\end{align}
	\end{small}
	Define $F(\phi)=\frac{(D^h)^2}{2}\frac{\sin\phi}{\cos\phi-\sin\phi\cot\psi}$. Since $F'(\phi)=\frac{l^2(\phi)}{2}$, we have}
	{	\setlength\abovedisplayskip{0.2cm}
		\setlength\belowdisplayskip{0.1cm}
	\begin{small}
	\begin{align}
	\int_{0}^{\phi_l} \frac{l^2(\phi)}{2}d\phi+\int_{\phi_u}^{2\pi} \frac{l^2(\phi)}{2}d\phi=\frac{\sin(\phi_l-\phi_u)l(\phi_l)l(\phi_u)}{2},
	\end{align}
	\end{small}
	which ends the proof.}
	%\vspace{-0.2cm}
\end{proof}
%\begin{align}
%\label{area}
%S&=\int_{0}^{2\pi} \frac{1}{2}c^2(\phi)d\phi\notag\\
%&=\int_{\phi_l}^{\phi_u} \frac{1}{2}d_{th}^2(\phi)d\phi+\frac{\sin(\phi_l-\phi_u)l(\phi_l)l(\phi_u)}{2}.
%\end{align}
\section{RIS Placement Optimization}
\label{sec_placement}
In this section, we first formulate a coverage maximization problem. A coverage maximization algorithm~(CMA) is then proposed to solve the formulated problem.
%\vspace{-0.4cm}
\subsection{Coverage Maximization Problem Formulation}
%\vspace{-0.1cm}
Since the horizontal distance $D^h$ between the RIS and the BS, and the orientation $\psi$ of the RIS will influence the cell coverage, our aim is to maximize the area of the cell coverage $S$ by jointly optimizing $(D^h,\psi)$, as formulated below:
%\begin{subequations}\label{opt_pro}
%\begin{align}
%\label{main}
%&\max_{H_R,D^h,\psi} S\\
%%\label{cons_SNR}
%s.t.&~\text{(\ref{above_thr})}.\notag
%\end{align}
%\end{subequations}
{\setlength\abovedisplayskip{0.1cm}
	\setlength\belowdisplayskip{0cm}
\begin{align}
\label{opt_pro}
\begin{split}
\max_{D^h,\psi}&~S~~~s.t.~\text{(\ref{above_thr})}.
\end{split}
\end{align}

%Constraint (\ref{cons_SNR}) ensures that the average SNR $\gamma$ is always above the threshold $\gamma_{th}$ when the horizontal distance $d_{BU}^h$ is smaller than $D^h$. 
%%
%\vspace{-0.5cm}
\subsection{Coverage Maximization Algorithm Design}
%\vspace{-0.1cm}
To solve the formulated problem, we first find the optimal RIS orientation for any horizontal distance $D^h$. The result is given by the following theorem.}
\begin{theorem}
	\label{opt_psi}
	Regardless of the horizontal distance, the optimal RIS orientation is always $\psi=\frac{\pi}{2}$, i.e., the RIS is deployed vertical to the direction from the BS to the RIS.
\end{theorem}
%\vspace{-0.2cm}
\begin{proof}
	See Appendix~\ref{app_opt_psi}.
	%\vspace{-0.1cm}
\end{proof}
Then, we optimize the horizontal distance $D^h$ given the optimal RIS orientation $\psi$, i.e.,
\begin{align}
\label{hor_dist}
\max_{D^h}&~S~~~s.t.~\text{(\ref{above_thr})}.
\end{align}
It can be found that $D^h$ has an influence on  $\phi_l$, $\phi_u$, and $d_{th}(\phi)$ in the objective function. However, it is hard to depict the influence in a closed-form expression. Inspired by the technique for solving the max-min problem~\cite{SL_2004}, by regarding $\phi_l$, $\phi_u$, and $d_{th}(\phi)$ as additional optimization variables, their closed-form expressions are no longer needed. However, due to the integral in the objective function, an infinite number of additional variables are needed. To cope with this challenge, we discretize the integral in the objective function as below,
%{\setlength\abovedisplayskip{0cm}
%	\setlength\belowdisplayskip{0cm}
\begin{align}
\int_{\phi_l}^{\phi_u} \frac{1}{2}d_{th}^2(\phi)d\phi\approx \sum_{i=0}^{K-1}\frac{1}{2}d_{th}^2(\phi_l+i\Delta) \Delta,
\end{align}
where the integral interval is divided equally into $K$ parts, each of which has a dimension of $\Delta=\frac{\phi_u-\phi_l}{K}$. For simplicity, we use $y_i$ to replace $d_{th}(\phi_l+i\Delta)$. Therefore, horizontal distance optimization problem (\ref{hor_dist}) can be transformed into its equivalent form,
\begin{subequations}\label{opt_pro_v3}
\begin{align}
\max_{\substack{\phi_l,\phi_u,D^h\\\{y_0,\dots,y_K\}}}& \sum_{i=0}^{N-1}\frac{1}{2}y_i^2 \Delta+\frac{\sin(\phi_l-\phi_u)l(\phi_l)l(\phi_u)}{2}\\
\label{cons_d_th}
s.t.~&g(\phi_l+i\Delta,y_i)=\gamma_{th}, i=0,\dots,K,\\
%\label{cons_x_l}
%&y_0=l(\phi_l), \phi_l <\pi,\\
%\label{cons_x_u}
%&y_{K}=l(\phi_u), \phi_u >\pi.
&\text{(\ref{above_thr}), (\ref{cons_phi_l}), and (\ref{cons_phi_u}),}\notag
\end{align}
\end{subequations}
where constraint (\ref{cons_d_th}) is derived based on (\ref{d_th}). Denote an optimal solution for problem (\ref{opt_pro_v3}) by $(\phi_l^*,\phi_u^*,(D^h)^*,\bm{y}^*)$, where $\bm{y}^*=\{y_0^*,\dots,y_K^*\}$. We have $\big((D^h)^*,\frac{\pi}{2}\big)$ is optimal for RIS deployment problem (\ref{opt_pro}). This is because problem (\ref{opt_pro_v3}) is equivalent to (\ref{hor_dist}). Therefore, $(D^h)^*$ is optimal for horizontal problem (\ref{hor_dist}). Besides, in (\ref{hor_dist}), the horizontal distance is optimized given the RIS orientation as $\frac{\pi}{2}$, and thus, $((D^h)^*,\frac{\pi}{2})$ is an optimal solution for RIS deployment problem (\ref{opt_pro}).
%Constraints (\ref{cons_x_l}) and (\ref{cons_x_u}) correspond to (\ref{cons_phi_l}) and (\ref{cons_phi_u}), respectively.

Problem (\ref{opt_pro_v3}) is a constrained optimization problem, which can be solved by the interior point method~\cite{SL_2004}. The basic idea can be summarized as follows: in each iteration, we approximate problem (\ref{opt_pro_v3}) by removing  inequality constraints in (\ref{opt_pro_v3}), and adding an additional logarithmic barrier function consisting of the inequality constraint functions to the objective function of (\ref{opt_pro_v3}). As such, the approximated problem is then solved utilizing the Newton method.

\section{Simulation Results}
\label{sec_simulation}
%\vspace{-0.05cm}
In this section, we verify our theoretic results on the cell coverage by simulations. The parameters are based on the 3GPP standard~\cite{38901} and existing work~\cite{BHLYZH_2020}, which are summarized in Table~\ref{notation}. To show the effectiveness of the proposed CMA, we compare it with another two algorithms, i.e., random algorithm and BS side scheme~(BSS). In the random algorithm, the horizontal distance and RIS orientation are determined in a random manner. Based on the conclusion that the received SNR at the receiver will be maximized when the RIS is close to the receiver or transmitter~\cite{HBLZ_2020}, the RIS is placed close to the BS at the edge of the far field in scheme BSS.
\begin{table}[!tpb]
	\small
	\centering
	\caption{\normalsize{Simulation Parameters}}
	\label{notation}
	\begin{tabular}{|l|l|}	
		\hline	
		\textbf{Parameters} & \textbf{Values} \\
		\hline\hline
		Height of the BS $H_B$ & $35$~m \\
		\hline
		Height of the UE $H_U$ & $1.5$~m \\
		\hline
		Height of the RIS $H_R$ & $2$~m \\
		\hline
		Size of the RIS element $s_M=s_N$ & $0.04$~m \\
		\hline
		Noise power $\sigma^2$ & $-96$~dBm \\
		\hline
		Wavelength $\lambda$ & $0.1$~m \\
		\hline
		Antenna gain $G$ & $1$ \\
		\hline
		Pathloss exponent $\alpha$ & $2$ \\
		\hline
		UE sensitivity $\gamma_{s}$ & $8$~dB \\
		\hline
		Margin for penetration loss $L_{mar}$ & $28$~dB \\
		\hline
		SNR threshold $\gamma_{th}$ & $36$~dB \\
		\hline
		Discretization level $K$ & $50$ \\
		\hline
	\end{tabular}
\end{table}
% Specifically, the height of the BS, the UE, and the RIS are selected as $H_B=35$~m, $H_U=1.5$~m, and $H_R=2$~m, respectively. The size of the RIS element is $s_M=s_N=0.04$~m. The noise power is $\sigma^2=-96$~dBm. The wavelength is set as $\lambda=0.1$~m, i.e., the carrier frequency is $3$~GHz. The antenna gain is given by $G=1$. The pathloss exponent is $\alpha=2$, and the attenuation is $\xi=28$~dB. The SNR threshold is set as $\gamma_{th}=8$~dB. The discretization level is set as $K=50$.
%	{\setlength{\abovecaptionskip}{0.1cm}
%	\setlength{\belowcaptionskip}{0cm} 

%}
\begin{figure*}[!tpb]
	\centering
	\begin{minipage}[b]{0.3\textwidth}
		\centering
		%\vspace{-0.4cm}
		\includegraphics[width=1\textwidth]{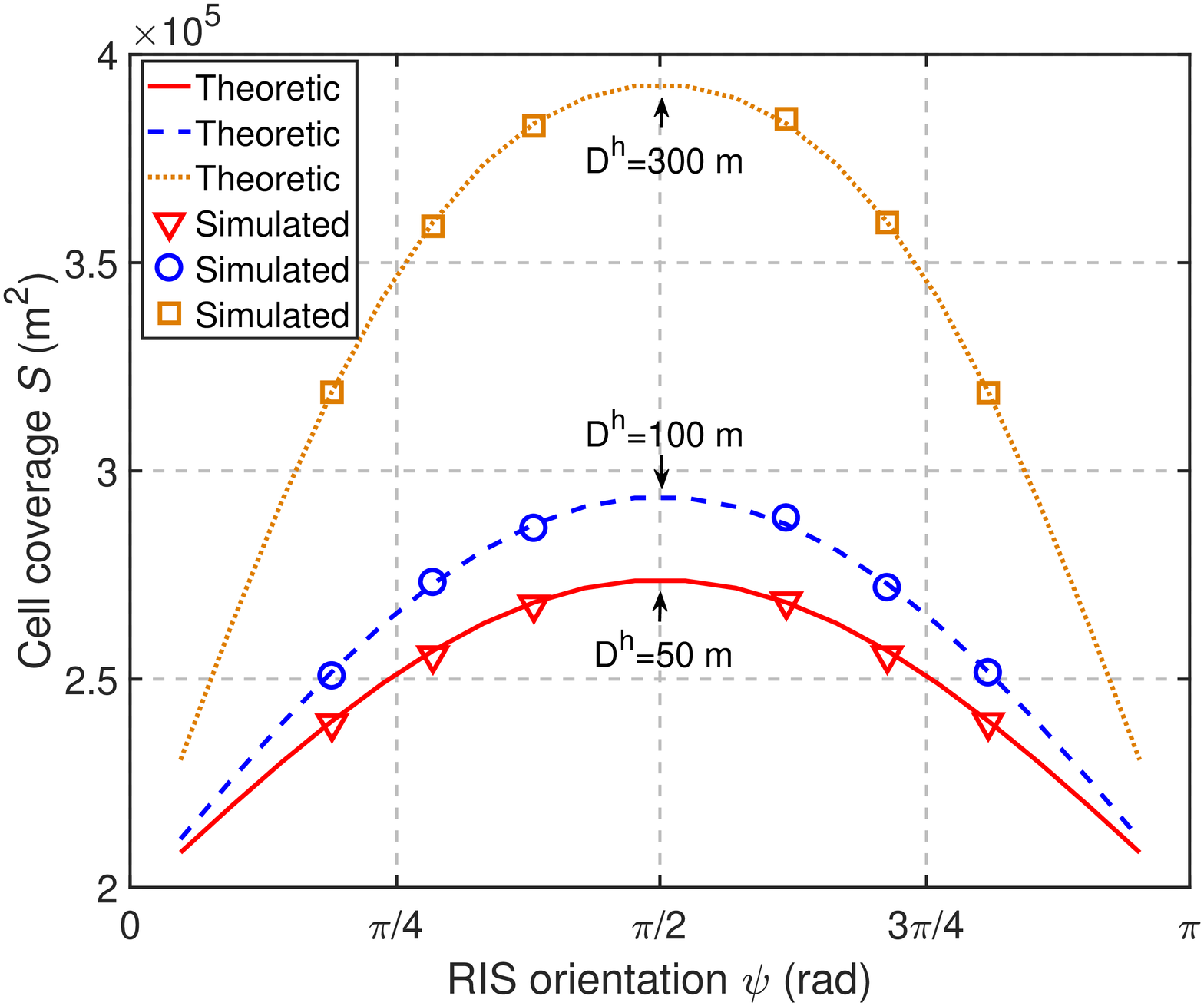}
		%\vspace{-0.7cm}
		\caption{Cell coverage vs. RIS orientation, \newline with $P=2$~W and $M=N=25$.}
		%				\vspace{-0.3cm}
		\label{coverage}
	\end{minipage}
	\begin{minipage}[b]{0.3\textwidth}
		\centering
		%\vspace{-0.1cm}
		\includegraphics[width=1\textwidth]{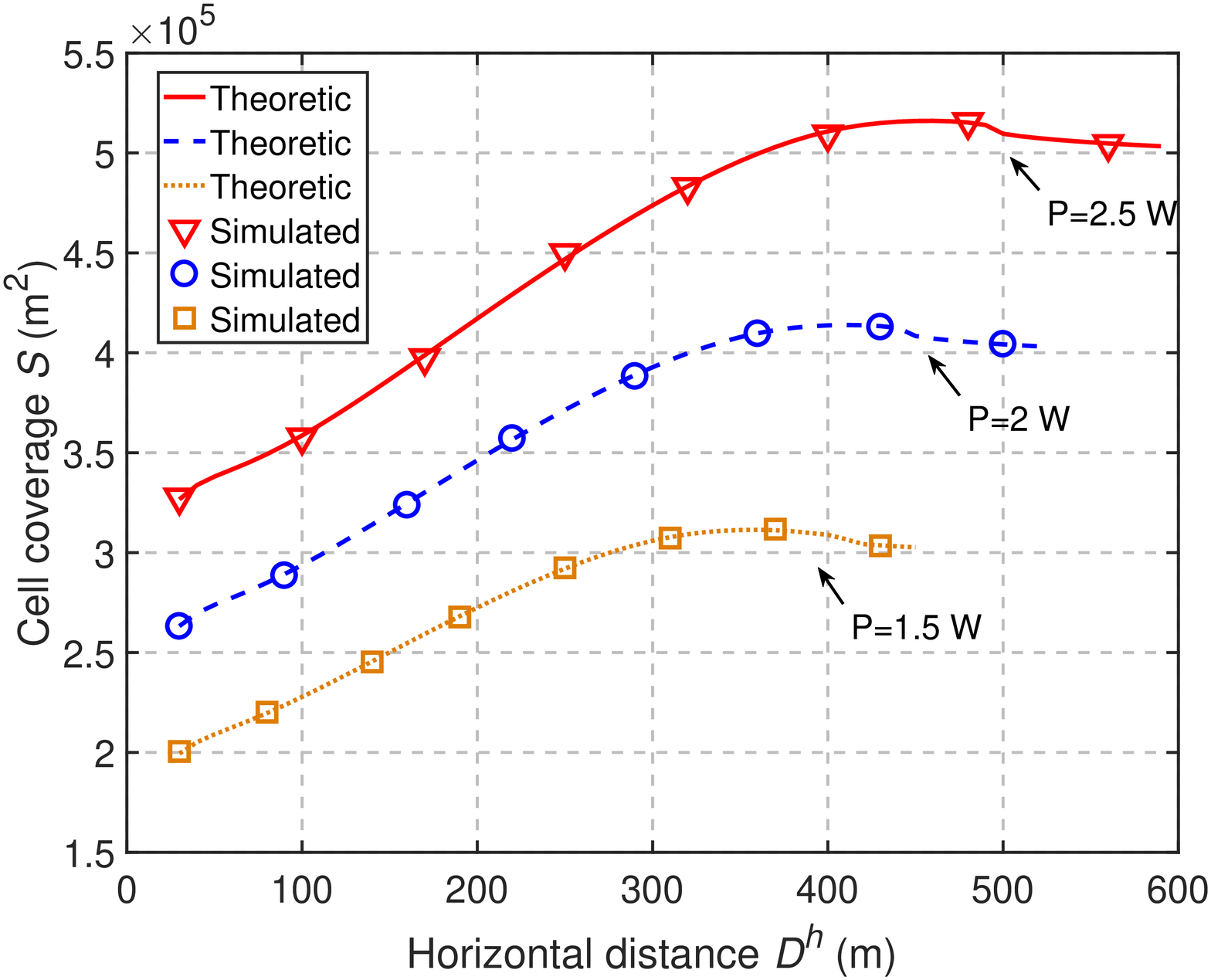}
		%\vspace{-0.7cm}
		\caption{Cell coverage vs. horizontal distance, \newline with $\psi=\frac{\pi}{2}$ and $M=N=25$.}
		%				\vspace{-0.3cm}
		\label{cvh}
	\end{minipage}
	\begin{minipage}[b]{0.3\textwidth}
		\centering
		%\vspace{-0.1cm}
		\includegraphics[width=1\textwidth]{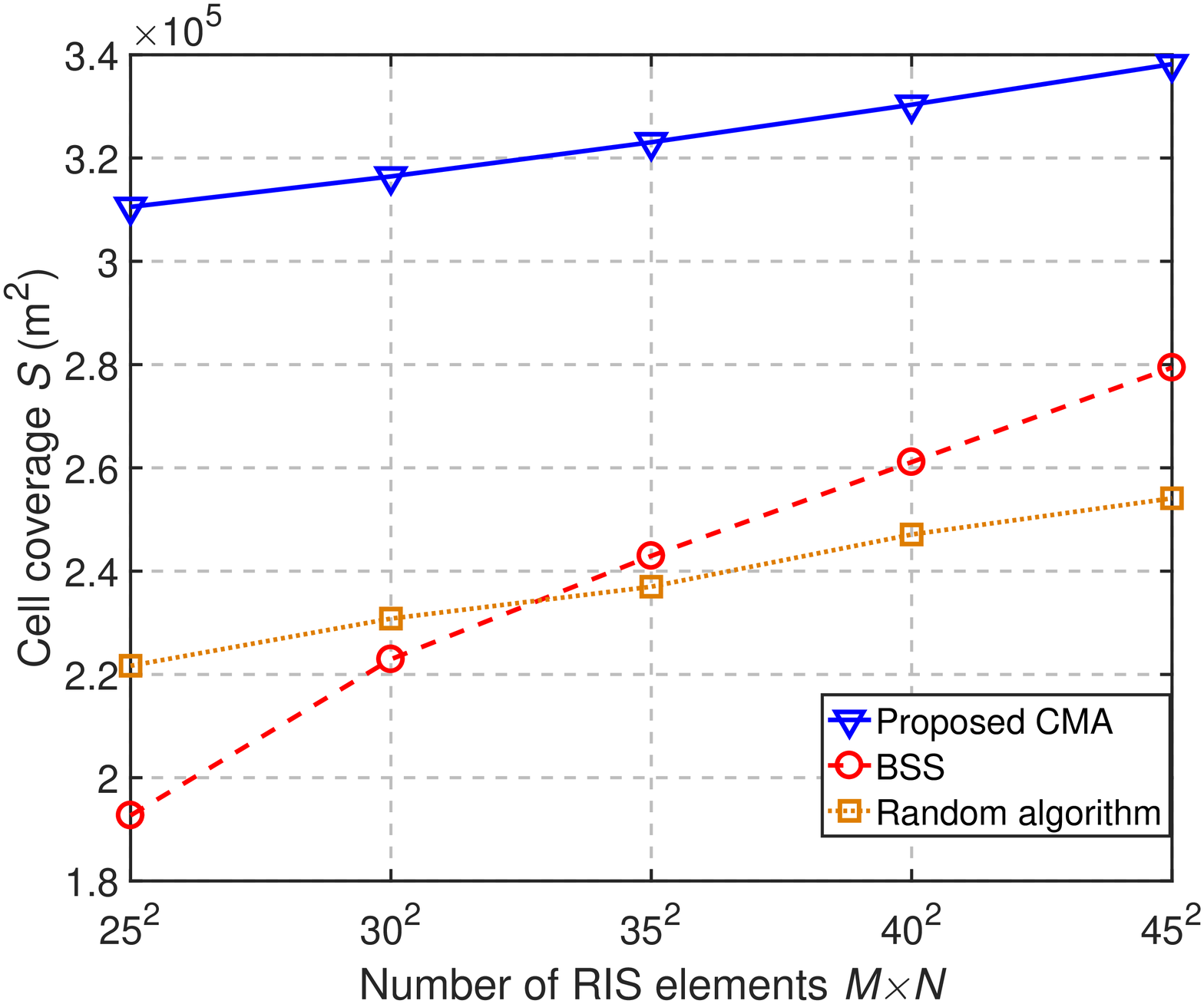}
		%\vspace{-0.7cm}
		\caption{Cell coverage vs. number of RIS \newline elements, with $P=1.5$~W}
		%				\vspace{-0.3cm}
		\label{cvr}
	\end{minipage}
	%\setlength{\belowcaptionskip}{-0.1cm}
	%\vspace{-0.4cm}
\end{figure*}
Fig.~\ref{coverage} shows the cell coverage $S$ versus RIS orientation $\psi$. The theoretic result is based on (\ref{area}), and the simulated result is obtained by $10^5$ Monte Carlo simulations. We can observe that the theoretic and simulated results match well, which verifies our derivation. From Fig.~\ref{coverage}, we can also find that the optimal RIS orientation is $\psi=\frac{\pi}{2}$, which is consistent with Theorem~\ref{opt_psi}.

Fig.~\ref{cvh} depicts the cell coverage $S$ versus the horizontal distance $D^h$ under the optimal RIS orientation. We can observe that the theoretic and simulated results are consistent, which further verifies our derivation. From Fig.~\ref{cvh}, we can also find that when the RIS moves away from the BS, the coverage will become larger first since the reflection coefficient increases. However, when $D^h$ further increases, the cell coverage degrades since the received SNR is negatively correlated with the distance between the BS and the RIS. Therefore, the RIS should be placed at a moderate distance from the BS. Besides, it can be found that the optimal RIS deployment is close to the cell edge, and thus, the RIS can significantly benefit UEs at the cell edge. According to Fig.~\ref{cvh}, it can also be observed that the cell coverage increases with the transmit power, due to improvement of the received SNR.
% When $D^h$ continues to grow, the cell coverage will saturate, since the RIS is far away from the BS and the UE such that the RIS has little influence on the cell coverage.

%\begin{figure}[!tpb]
%%	\setlength{\abovecaptionskip}{-0.1cm}
%%	\setlength{\belowcaptionskip}{-0.2cm}  
%	\centering
%	\vspace{-0.2cm}
%	\includegraphics[width=2.8in]{coverage_vs_height.eps}
%	\vspace{-0.2cm}
%	\caption{Cell coverage vs. RIS height.}
%	\vspace{-0.3cm}
%	\label{cvr}
%\end{figure}

Fig.~\ref{cvr} shows the cell coverage $S$ versus the number of RIS elements $MN$. From Fig.~\ref{cvr}, we can find that the proposed CMA achieves larger cell coverage than the random algorithm. Besides, it can be observed that the CMA outperforms the BSS, since the influence of incidence angle on the reflection coefficient is not considered in~\cite{HBLZ_2020}. According to Fig.~\ref{cvr}, when the number of RIS elements increases, the cell coverage becomes larger, since the received SNR is positively correlated with the number of elements.
\section{Conclusion}
\label{sec_conclusion}
In this letter, we have considered a downlink RIS-assisted network with one BS and one UE. The coverage of this network has been analyzed, and a problem has been formulated to maximize the cell coverage by optimizing the RIS placement. To solve the formulated problem, we have proposed the algorithm CMA. From the analysis and simulation, we can conclude that the RIS should be deployed vertical to the direction from the BS to the RIS. Besides, we should place the RIS at a moderate distance from the BS.
%\vspace{-0.2cm}
\begin{appendices}
 	\vspace{-0.2cm}
% 	\section{Proof of Remark~\ref{d_A}}
% 	\label{app_region_A}
%% 	\vspace{-0.1cm}
%% 	We first show that $\gamma$ is always above $\gamma_{th}$ when $d_{BU}^h \le D^h$. Specifically, according to (\ref{max_SNR_cont}) and (\ref{above_thr}), we have
%% 	\begin{align}
%% 	\gamma \ge \eta_D d_{BU}^{-2} \ge \eta_D (\sqrt{(D^h)^2+(H_B-H_U)^2})^{-2}\ge \gamma_{th}.\notag
%% 	\end{align}
%% 	\vspace{-0.1cm}
%
% 	When $d_{BU}^h>D^h$, the increase of $d_{BU}^h$ will lead to larger $d_{BU}$ and $d$. Since $\gamma$ strictly decreases with $d_{BU}$ and $d$, $\gamma$ strictly degrades with the horizontal distance $d_{BU}^h$. Besides, it is easy to find that when $d_{BU}^h\rightarrow\infty$, $\gamma$ approaches to $0$.
% 	
% 	
% 	\vspace{-0.3cm}
 	\section{Proof of Remark \ref{profile}}
 	\label{app_profile}
 	\vspace{-0.1cm}
 	First, we show that as $\phi$ approaches to $\pi$, $d_{th}(\phi)$ will become smaller while $l(\phi)$ will be larger. Based on cosine theorem, the horizontal distance between the UE and the RIS can be expressed as $d^h=\sqrt{(D^h)^2+(d_{BU}^h)^2-2D^hd_{BU}^h\cos(\phi)}$.
% 	\begin{align}
% 	\label{cos}
% 	d^h=\sqrt{(D^h)^2+(d_{BU}^h)^2-2D^hd_{BU}^h\cos(\phi)}
% 	\end{align}
 	Therefore, when $\phi$ approaches to $\pi$, $d^h$ will become larger, which leads to a larger distance between the UE and the RIS, and thus, a lower SNR $\gamma$. Therefore, $d_{th}(\phi)$ becomes smaller. 
 	
 	When $\phi=\psi$ or $\phi=\psi+\pi$, $l(\phi)\rightarrow\infty$ while $d_{th}(\phi)$ is limited. On the contrary, when $\phi=0$, we have $l(\phi)=D^h<d_{th}(\phi)$. Therefore, there are two angles satisfying $d_{th}(\phi)=l(\phi)$, and they are located within $[0,\psi)$ and $(\psi+\pi,2\pi)$, respectively. Besides, $d_{th}(\phi)$ is larger than $l(\phi)$ only when $\phi \in [0,\phi_l)\cup(\phi_u,2\pi)$.
 
\vspace{-0.3cm} 

\section{Proof of Theorem~\ref{opt_psi}}
\label{app_opt_psi}
\vspace{-0.1cm}
\begin{figure}[!tpb]
	\centering 
	\subfigure[]{
		\label{fig_ilu_1}
		\includegraphics[width=2.3in]{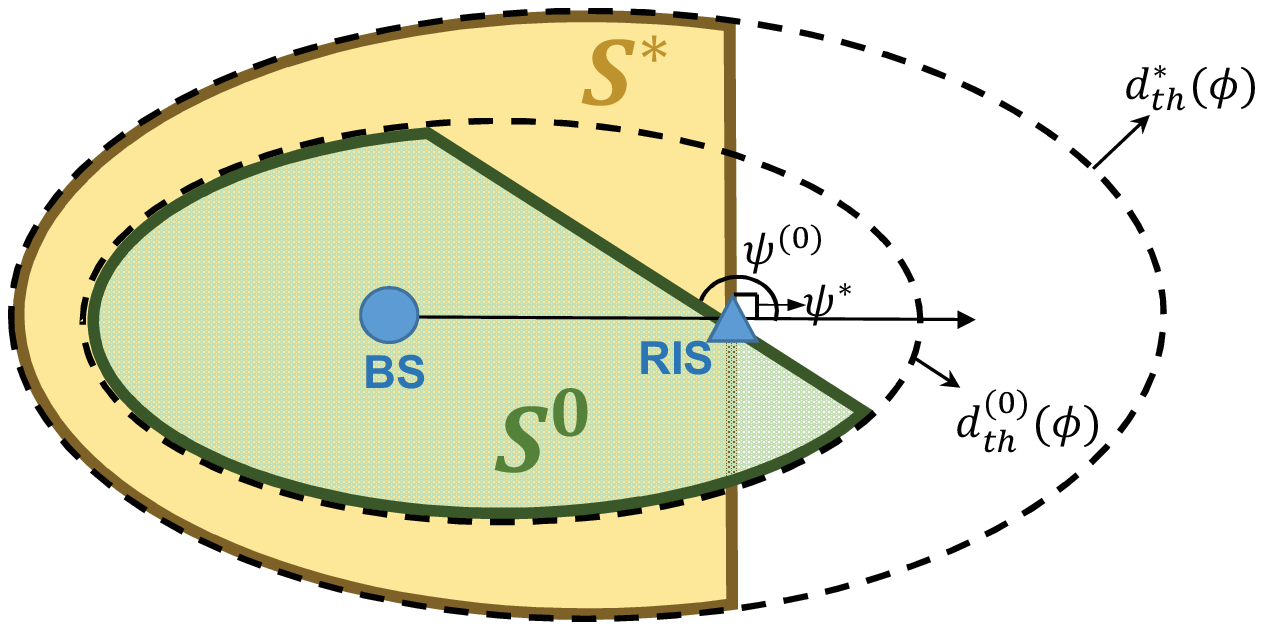}}
	\subfigure[]{
		\label{fig_ilu_2}
		\includegraphics[width=2.2in]{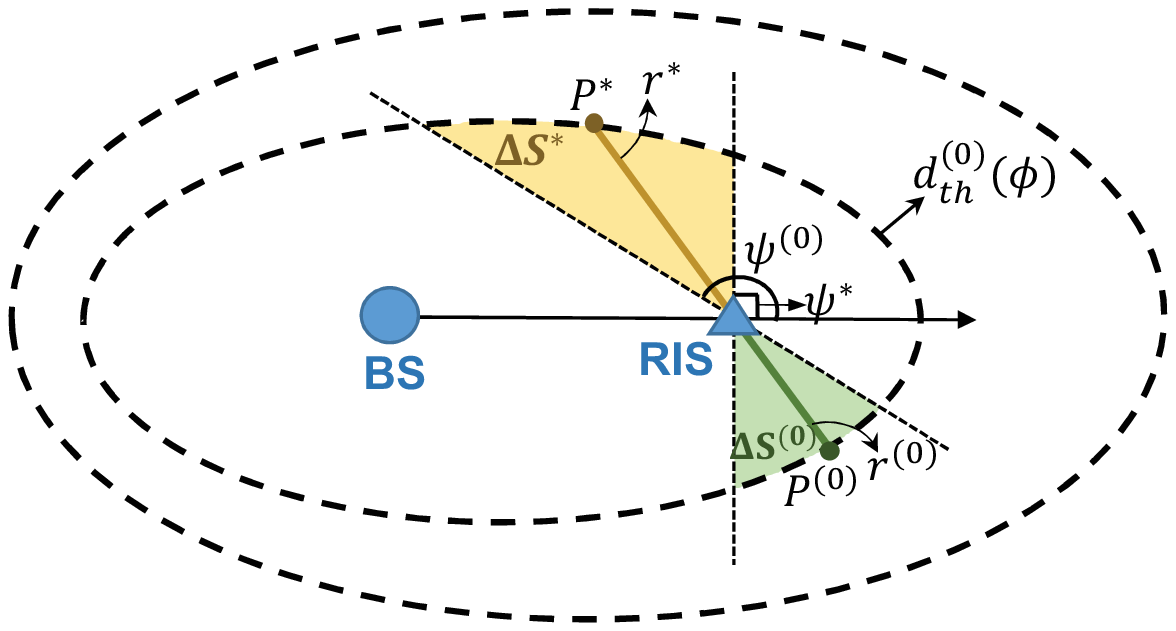}}
	\caption{Graphical interpretation for proof of Theorem~\ref{opt_psi}.}
	\label{fig_ilu}
\end{figure}
To show that $\psi^*=\frac{\pi}{2}$ is optimal, we compare $\psi^*$ with another RIS orientation $\psi^{(0)}$. As shown in Fig.~\ref{fig_ilu_1}, the cell coverage under $\psi^*$ and that under $\psi^{(0)}$ correspond to the yellow region and green region, respectively, where the area of the two regions are denoted by $S^*$ and $S^{(0)}$, respectively. It is worthwhile noting that $d_{th}^*$ is larger than $d_{th}^{(0)}$ since as $\psi$ deviates from $\frac{\pi}{2}$, reflection amplitude $\Gamma$ decreases, and thus, the SNR at the UE degrades. In the following, we show that $S^{(0)}\le S^* $, from which we can obtain that $\frac{\pi}{2}$ is the optimal RIS orientation.

Note that $S^*-S^{(0)}>\Delta S^*-\Delta S^{(0)}$, where $\Delta S^*$ and $\Delta S^{(0)}$ are the areas of the two sectors shown in Fig.~\ref{fig_ilu_2}. Therefore, we will focus on comparing $\Delta S^*$ with $\Delta S^{(0)}$ in the following. To compare $\Delta S^*$ with $\Delta S^{(0)}$, we arbitrarily select two opposite radius of the sectors, denoted by $r^*$ and $r^{(0)}$, respectively, and we use $P^*$ and $P^{(0)}$ to represent the intersection between $r^*$ and $d_{th}^{(0)}$, and that between $r^{(0)}$ and $d_{th}^{(0)}$, respectively. Compared with the case where the UE is in $P^{(0)}$, the UE is closer to the BS when it is in $P^*$. However, SNR $\gamma$ is equal to $\gamma_{th}$ when the UE is in $P^{(0)}$ or $P^*$. Therefore, when the UE is in $P^{(0)}$, it is closer to the RIS compared with the case where the UE is in $P^*$, i.e., $r^{(0)} <r^*$. As a result, we have $\Delta S^{(0)}<\Delta S^*$, which indicates that $S^*\ge S^{(0)}$.
\end{appendices}
%\vspace{-0.1cm} 
%%%%%%%%%%%%%%%%%%%%%%%%%%%%%%%%%%%%%%%%%%%%%%%

\end{document}